\documentclass[10pt,twocolumn]{IEEEtran}
\usepackage{amsmath,amssymb,graphicx,epsfig}
\usepackage[tight,footnotesize]{subfigure}
\usepackage{textcomp}

\newtheorem{theorem}{Theorem}[section]

\newtheorem{proposition}[theorem]{Proposition}
\newtheorem{assumption}[theorem]{Assumption}

\def\bx{{\bf x}}
\def\bb{{\bf b}}
\def\bc{{\bf c}}

\def\beq{\begin{equation}}
\def\eeq{\end{equation}}

\title{{A Scheduling Model of Battery-powered Embedded System}}

\author{ Zhenwu Shi\\
      School of Electrical and Computer Engineering \\
      Georgia Institute of Technology\\
      Atlanta, GA, 30332 \\
      Email: \{zwshi\}@gatech.edu}
\begin{document}
\maketitle

\begin{abstract}
Fundamental theory on battery-powered cyber-physical systems (CPS) calls for dynamic models that are able to describe and predict the status of processors and batteries at any given time. We believe that the idealized system of single processor powered by single battery (SPSB) can be viewed as a generic case for the modeling effort.
This paper introduces a dynamic model for
multiple aperiodic tasks on a SPSB system under a scheduling  algorithm that resembles the rate monotonic scheduling (RMS) within finite time windows.
The model contains two major modules.
The first module is an online battery capacity model based on the Rakhmatov-Vrudhula-Wallach (RVW) model. This module provides predictions of remaining battery capacity based on the knowledge of the battery discharging current. The second module is a dynamical scheduling model that can predict the scheduled behavior of tasks within any finite time window, without the need to store all  past information about each  task before the starting time of the finite time window. The module provides a complete analytical description of the relationship among tasks and it delineates all possible modes of the processor utilization as square-wave functions of time. The two modules i.e. the scheduling model and the battery model are integrated to obtain a
hybrid scheduling model that describes the dynamic behaviors of the SPSB system. Our effort may have demonstrated that through dynamic modeling, different components of CPS may be integrated under a unified theoretical framework centered around hybrid systems theory.
\end{abstract}

\section{Introduction}
Modern society is characterized by the pervasive applications of
embedded computing systems powered by batteries. Use of batteries
endows computing systems with mobility that is appreciated by
consumers, industries and governments. Cell phones, portable
music/video players, and unmanned airplanes flying over a battle
zone are among the most noticeable applications. Computing systems
supported by batteries are usually designed following specialized
guidelines to reduce power consumption, achieving lower maintenance
and longer operation time.

Despite of the great achievements in battery-powered computing,
better theoretical understandings of the interactions between
computing systems and batteries will be appreciated by applications
where power is very limited or battery life is a dominant constraint
on systems design. One such application is in underwater robotics,
where missions may last more than a year. Although appear to be
simple, batteries are complex physical-chemical systems.  Hence
fundamental research in battery powered embedded computing systems
may be viewed as part of the evolving cyber-physical systems (CPS)
theory (c.f. recent publications
\cite{HuLemmon06,XiaSun06,seto01,AP12,Kottenstette08,ZWS_RTSS08})
that aims to develop novel foundations to understand complex
physical systems controlled by modern computers.

The complexity of battery discharging behaviors is noticed in
literature e.g. \cite{Chia01, Rao03}. Battery modeling aims to
simulate these behaviors by computational models \cite{CPMT2012, MTT2010}. To
support theoretical cyber-physical systems design, a class of
analytical models are desired. This is because comparing to physical
\cite{Dolye93}, empirical \cite{Linden01}, and circuit based
\cite{Chen06, EPEPS2011} models, analytical models are theoretically tractable
while providing sufficient accuracy.

The interactions between batteries and computing systems have been
studied by some researchers. Power consumption of computing devices
is proportional to $V^2_{\rm DD}$ \cite{Rabaey96}, hence can be
reduced by dynamic voltage scaling (DVS). Yao et al. \cite{Yao95}
propose one of the first DVS-aware scheduling algorithms.  Rowe et.
al. \cite{Rajkumar08} proposed a scheduling method to combine idle
period of tasks so that the processor can be put to sleep to save
power. Some recent results on power management in sensor network
applications are presented by Ren et. al. \cite{Krogh05}. In these
works the behaviors of batteries are  simplified as ideal voltage
sources whose life only depends on the average discharge current.
The approach taken by Jiong et. al. \cite{JiongJha07a, JiongJha07b}
uses an empirical battery model and  assumes known battery discharge
profiles for computing tasks.

We believe that a dynamic model needs to be established that is able
to predict battery capacity based on the discharge current
determined by a feedback control law or feedback scheduling
algorithm. The magnitude and pulse width of such discharge current
can not be determined beforehand. Therefore, in previous work
\cite{ZSW_WCPS09}, we introduce a dynamic battery model that
describes the variations of the  capacity of a battery under time
varying discharge current. This model is input-output equivalent to
the Rakhmatov-Vrudhula-Wallach (RVW) model proposed  by Rakhmatov
et. al. \cite{Rakhmatov03} which has been  shown to agree with
experimental results and has demonstrated high accuracy in battery
capacity prediction and battery life estimation. Besides the accuracy, the RVW model is also simple to use. It can characterize any general battery with only two parameters $\alpha$ and $\beta$, which can be estimated by a experimental method introduced in \cite{Rakhmatov03}. However, the
application of the RVW model requires that the discharge current, as
a function of time, is known. Our dynamic model allows battery
capacity prediction for feedback control laws and online scheduling
algorithms.

This paper is inspired by the need  to completely describe the possible battery
discharge profiles for a processor that is running a finite number of aperiodic
real-time tasks. Knowing such a discharge profile allows us to
predict remaining battery capacity at any given moment without seeking the help from
simulation programs. We found that
the problem can only be solved by establishing a dynamic model that
integrates both battery models and task scheduling models. However, we are
unable to find existing models from the literature that fits in such requirements.

Therefore, we start our modeling effort by considering a single processor single battery (SPSB) model.
 This model may
appear naive. However, we find it is very challenging to establish
a complete analytical description for the behaviors of multiple tasks. Readers in doubt about this
difficulty may consider giving an explicit mathematical formula for
the processor utilization plots generated by real-time systems
simulation tools such as the TrueTime \cite{Cervin03}. Such
difficulty has been mostly avoided by the works on schedulability
\cite{LiuLayland73,Lehoczky89,Bini03} because only the worst case
scenario, which are the critical time instances, are considered.

With the help from hybrid systems theory evolving in the control
community \cite{Antsaklis00,Branicky98}, we obtain a class of hybrid
dynamical systems models that are able to describe the battery
discharge profile when multiple independent aperiodic tasks are
scheduled under a fixed priority scheduling algorithm on a single processor. When
integrated with the dynamic battery model, theoretical predictions
of battery capacity for an SPSB system can be produced online at any
given time. Comparing to the literature reviewed, we believe our
contribution is novel and may provide foundation  for
developing novel battery aware scheduling for cyber-physical systems.

The paper is organized as follows. Section \ref{formulation}
describes the  SPSB system we investigate. Section
\ref{battery} introduces a dynamic battery model that is used to
determine the remaining battery capacity. Section \ref{Model}
derives the dynamic models for two aperiodic tasks on a single
processor. The dynamic model also enables us to find analytic
descriptions of the processor utilization waveforms in Section
\ref{CPU}. Section \ref{multiple} extend the conclusion of two tasks to multiple tasks. The models are verified and demonstrated by simulations
in Section \ref{simulation}. Conclusions are provided in Section
\ref{con}.

\section{\label{battery}A Dynamic Battery Model}
In this section, we briefly review the dynamic battery model established in our recent work \cite{ZSW_WCPS09} as a workshop paper.
We model the discharge dynamics within one discharge cycle of
(possibly rechargeable) batteries. The discharge cycle starts when a battery is fully charged and ends when the battery can not support the demand for discharge current; no recharging is allowed during the cycle.
According to battery and VLSI design literature e.g. \cite{Chia01, Rao03},
the discharge current is supported by the change of concentration of electrolytes near the electrodes of a battery.  When the concentration {\it at the electrodes} drops below a certain threshold, a battery fails to support discharge current, causing failure to devices it supports. At this moment, the discharge cycle has to stop and the battery needs to be
 recharged or replaced. Note that there may be a significant amount of electrolytes left when a discharge cycle ends.
When the battery is discharged under a pulsed discharge current, during the idle time when current is interrupted, the diffusion process increases electrolyte concentration at the electrodes. This produces the {\it recovery effect} that makes the battery appears to have regained portions of its capacity.

The Rakhmatov-Vrudhula-Wallach (RVW) model in \cite{Rakhmatov03} is derived from solving the diffusion equations governing the electrolytes motion within a battery. The model captures the recovery effect effectively. As mentioned in the introduction, in order to use the RVW model for CPS design,
 we derived a dynamic model that is input-output equivalent to the RVW model. The key improvement is to introduce a set of state variables, denoted by $\bx= [x_0, x_1,..., x_m]$ where $m$ is chosen according to accuracy requirement. In most cases, $m\le 10$ is accurate enough.
Then the dynamic model is given by a linear system:
\begin{eqnarray}
        \dot \bx &=& A \bx + \bb i(t) \cr
        y & =& \bc' \bx.
\end{eqnarray}
 Here, the matrix $A$ is a diagonal matrix i.e. $A={\rm diag}\{0, -\lambda_1,...,-\lambda_m\}$ where $\lambda_j=\beta*j^2$
  for $j=1,2,...,m$. The input $i(t)>0$ is the discharge current. The vector $\bb=\frac{1}{\alpha}[1,2,2,...,2]'$ and the vector $\bc=[1,1,1,...,1]'$. As mentioned in the introduction, $\alpha$ and $\beta$ characterize a battery and they can be estimated by the experimental method introduced in \cite{Rakhmatov03}. The output variable $y$ measures the total capacity loss for a battery. It contains two parts, the permanent capacity loss represented by
  $x_0$, and the temporary capacity loss measured by $x_1, x_2, ..., x_m$.  At the
 beginning of the discharge cycle the initial condition for $\bx$ is the zero vector, hence $y=0$. When $y=1$, the discharge cycle ends.
 The first state variable $x_0$ is the integration of the discharge current $i(t)$, hence is the {\it effective} discharge delivered to the circuits outside the battery. One can see that the actual discharge $y$ contains more than $x_0$. The states $x_1, x_2,..., x_m$ are temporary capacity losses that are always positive since $i(t)\ge 0$.  Under an impulsive current, when $i(t)=0$, $x_j ( 1\le j\le m)$ decreases because $-\lambda_j<0$. This captures the recovery effect.
More details about this model and results regarding optimal discharge profiles are
presented in our work \cite{ZSW_WCPS09}.

\section{\label{formulation}Scheduled Behaviors of Aperiodic Tasks}
In this paper, we consider the case where aperiodic
tasks with hard deadlines scheduled on a SPSB system. Tasks are scheduled using the following scheduling
algorithm: the tasks with shorter
request intervals are assigned higher priorities, and the higher
priority tasks can preempt the lower priority tasks. {\bf We are
interested in the scheduled behaviors of $\tau_1, ..., \tau_N$ within
a finite time window, which starts from
$t_0$ with length $L$}. $t_0$ can be any time instant.
Since the request intervals of aperiodic tasks are not
constant, the priority assignments among aperiodic tasks may change during the finite time window. To avoid such cases, we enforce the following assumptions on the length of time window $L$.
\begin{assumption}
\label{assumption:1} Within $[t_0, t_0+L]$, the maximum
request interval of $\tau_{i}$ is smaller than the minimal
request interval of $\tau_{i+1}$, where $i$ is selected such that the priority of $\tau_i$ is always higher than the priority of $\tau_{i+1}$.
\end{assumption}

When $t_0$ changes, the length of time window and the priority assignments may change as well.

\begin{assumption}
\label{assumption:2}
We assume that for any time instant $t$,  there exists $t_0$ and $L$ such that Assumption \ref{assumption:1} is satisfied.
\end{assumption}

Assumption \ref{assumption:1} guarantees coherence in priority assignments during the finite time window starting from $t_0$. Assumption \ref{assumption:2} guarantees that there is no gap between the proper time windows. It is of course trivial to see that periodic tasks satisfy these assumptions within $[t_0, t_0+L]$. In fact, the scheduling algorithm becomes
the Rate Monotonic Scheduling (RMS) when the time window $L$ goes to infinity. Our main goal here is to model the behaviors of the tasks, not proposing a new scheduling algorithm.

Since such an modeling effort has not been performed in previous literature, we need to introduce new notations that depict the task behaviors. Suppose
$\tau_{\alpha}$ is an independent aperiodic task. We use $t_{\alpha}(n)$ to denote the request time of the $n$-th instance of $\tau_{\alpha}$; $C_{\alpha}(n)$ to denote the computing time of the
$n$-th instance of $\tau_{\alpha}$; and $T_{\alpha}(n)$ to denote the request interval of the $n$-th instance of $\tau_{\alpha}$, which is defined to be the interval between the request time of the $n$-th instance of $\tau_{\alpha}$, and the request time of the $(n+1)$-th instance of $\tau_{\alpha}$, i.e. we have
$T_{\alpha}(n)=t_{\alpha}(n+1)-t_{\alpha}(n)$.

Regardless of the choice of $t_0$, the scheduled behaviors of $\tau_1, ..., \tau_N$ within $[t_0, t_0+L]$ is affected by both the tasks requested within $[t_0, t_0+L]$ and the tasks requested before $t_0$. To predict the scheduled behaviors within $[t_0, t_0+L]$ by real-time simulation tools such as Truetime, we have to know all task information within $[0, t_0+L]$ and simulate from time $0$. However, keeping all task information is costly in the real application and simulating from the beginning is time inefficient. We discover that this problem can be solved by using a dynamical model, instead of storing a large amount of data, we just need to update the value of some state variables. Furthermore, the dynamic model is able to determine the state of
the processor at any given time within $[t_0, t_0+L]$.

We develop a recursive algorithm to determine the state of the processor for multiple aperiodic tasks. Our goal is to find out when the processor is occupied within $[t_0,t_0+L]$. The result is a function of time
$\Phi: t \to \{0,1\}$ where for $t\in [t_0, t_0+L]$,
$\Phi(t)=0$ if the processor is free, and $\Phi(t)=1$ if the
processor is busy.

Furthermore, we assume that when $\Phi(t)=1$, a constant current
will be drawn from the battery and when $\Phi(t)=0$, the current
vanishes. For simplicity, we ignore the possible power difference
among tasks and the transients in the battery discharge
current. If $\Phi(t)$ is known, we are able to
describe the current $i(t)$ drawn from the battery by the processor. And then, using the dynamic battery model,
we are able to predict the remaining battery capacity.

\section{\label{Model}A Dynamic Scheduling Model for Two Tasks}
In this section, we establish a dynamic model for the scheduled
behaviors of two independent aperiodic tasks on a single processor within
$[t_0, t_0+L]$. Later in this paper, we will develop a recursive algorithm to extend the model to
multiple tasks. Consider two tasks
$\tau_{\alpha}$ and $\tau_{\beta}$, where $\alpha=1$ and $\beta=2$, scheduled under the scheduling algorithm introduced in Section \ref{formulation}. Suppose Assumptions \ref{assumption:1} and \ref{assumption:2} are satisfied, we have the following claim:

{\bf Claim 1:\;}$\tau_{\alpha}$ is assigned higher priority than $\tau_{\beta}$ within $[t_0, t_0+L]$.

There is at most one instance of $\tau_{\beta}$ requested within each request interval of $\tau_{\alpha}$.

We determine a set of "state variables" that completely determine
the scheduled behavior of $\tau_{\alpha}$ and $\tau_{\beta}$ at any
given time within $[t_0, t_0+L]$. Then
transition rules between these states from the current time to all
future times are described by a set of nonlinear difference
equations.The work presented in this section centered around these
two themes---states and transition rules.

\subsection{Exploring Task Relationship}
In this subsection, we model the dynamics for the request time of
$\tau_{\alpha}$ and $\tau_{\beta}$, and explore task relationship between $\tau_{\alpha}$ and $\tau_{\beta}$.

We first select the request time of $\tau_{\alpha}$ and
$\tau_{\beta}$ as two state variables. Within the finite time window,
$t_{\alpha}(n)$ and $t_{\beta}(m)$ satisfies the following equations
\begin{equation}
\label{eq:t} \left\{
\begin{array}{l}
t_{\alpha}(n+1)=t_{\alpha}(n)+T_{\alpha}(n) \\
t_{\beta}(m+1)=t_{\beta}(m)+T_{\beta}(m)
\end{array}
\right.
\end{equation}
where $n$ and $m$ are the index numbers for the instances of $\tau_{\alpha}$ and $\tau_{\beta}$ requested within the finite time window. For example, if $\tau_{\alpha}$ has a fixed period, $n$ ranges from $1$ to $\lfloor L/T_1\rfloor$ where $T_1$ is the period for $\tau_{\alpha}$; and if $\tau_{\beta}$ has a fixed period $T_2$, $m$ ranges from $1$ to $\lfloor L/T_2\rfloor$.

Next, we explore task relationship between $\tau_{\alpha}$ and $\tau_{\beta}$ by introducing another two state variables $y(n)$ and $z(n)$
. $y(n)$ helps describe the request time of the first instance of $\tau_{\beta}$ requested after the $n$-th instance of $\tau_{\alpha}$, which is denoted by $t_{\beta}(y(n))$. The value of $y(n)$ indicates the index numbers for the instances of $\tau_{\beta}$ requested within the finite time window. If $\tau_{\beta}$ has fixed period $T_2$, then $y(n)$ ranges from $1$ to $\lfloor L/T_2\rfloor$. $z(n)$ describes the phase difference between the request time of the
$n$-th instance of $\tau_{\alpha}$, and the request time of the $y(n)$-th instance of $\tau_{\beta}$, i.e.
\begin{equation}
\label{eq:zdef} z(n)=t_{\beta}(y(n))-t_{\alpha}(n).
\end{equation}

To determine the transition rules for $y(n)$ and $z(n)$ from $n$ to $n+1$, we
need to consider two cases.

{\bf Case 1} This case happens when the phase difference satisfies
\begin{equation}
0 \le z(n)<T_{\alpha}(n)
\end{equation}

In this case, the $y(n)$-th instance of $\tau_{\beta}$ is requested within
$[t_{\alpha}(n), t_{\alpha}(n+1)]$. Thus, the $y(n)$-th instance of
$\tau_{\beta}$ is the last instance of $\tau_{\beta}$ requested before
$t_{\alpha}(n+1)$. On the other hand, since $y(n+1)$ is defined to
be the index number for the first instance of $\tau_{\beta}$ requested after
$t_{\alpha}(n+1)$, the $\{y(n+1)-1\}$-th instance of $\tau_{\beta}$ is also the last instance of
$\tau_{\beta}$ requested before $t_{\alpha}(n+1)$. Therefore, we have
\begin{equation}
\label{eq:y1} y(n+1)-1=y(n)
\end{equation}

In this case, $z(n)$ will update in the following ways
\begin{eqnarray*}
z(n+1)&=&t_{\beta}(y(n+1))-t_{\alpha}(n+1) \\
&=&t_{\beta}(y(n)+1)-t_{\alpha}(n+1) \\
&=&z(n)+T_{\beta}(y(n))-T_{\alpha}(n)
\end{eqnarray*}
where (\ref{eq:t}), (\ref{eq:zdef}) and (\ref{eq:y1}) have been
applied. This implies that the phase difference between the two
tasks has increased by $T_{\beta}(y(n))-T_{\alpha}(n)$.

{\bf Case 2} This case happens when the phase difference satisfies
\begin{equation*}
z(n)\ge T_{\alpha}(n)
\end{equation*}

In this case, the $y(n)$-th instance of $\tau_{\beta}$ is requested after
$t_{\alpha}(n+1)$. Thus, the $y(n)$-th instance of $\tau_{\beta}$ is the first instance of
$\tau_{\beta}$ requested after $t_{\alpha}(n+1)$. On the other hand,
the $y(n+1)$-th instance of $\tau_{\beta}$ is also defined to be the first instance of
$\tau_{\beta}$ requested after $t_{\alpha}(n+1)$. Therefore, we have
\begin{equation}
\label{eq:y2} y(n+1)=y(n)
\end{equation}

In this case, $z(n)$ will update in the following ways
\begin{eqnarray*}
z(n+1)&=&t_{\beta}(y(n+1))-t_{\alpha}(n+1) \\
&=&t_{\beta}(y(n))-t_{\alpha}(n+1) \\
&=&z(n)-T_{\alpha}(n).
\end{eqnarray*}
It can be seen that the phase difference decreases by
$T_{\alpha}(n)$.

In summary, $y(n)$ and $z(n)$ are described by the following
equations:
\begin{equation}
\label{eq:y} y(n+1)=\left\{
\begin{array}{ll}
y(n)+1 &  {\rm if} \;\; 0\leq z(n) < T_{\alpha}(n) \\
y(n)   &  {\rm if}  \;\;z(n) \ge T_{\alpha}(n)
\end{array}
\right.
\end{equation}

\begin{equation}
\label{eq:z} z(n+1)=\left\{
\begin{array}{l}
z(n)+T_{\beta}(y(n))-T_{\alpha}(n) \\
 \{{\rm if} \;\; 0 \le z(n)<T_{\alpha}(n)\} \\
z(n)-T_{\alpha}(n)  \\
 \{{\rm if} \;\; z(n)\ge T_{\alpha}(n)\}
\end{array}
\right.
\end{equation}

\subsection{Modeling the residue}
The state of processor after a given time point will be affected by the instances of $\tau_{\alpha}$ and $\tau_{\beta}$ requested before the given time point. To track the status of $\tau_{\alpha}$ and $\tau_{\beta}$, we introduce a term
called the {\it residue}. The residue $P(n)$
is defined to be the amount of time that is needed after the time instant
$t_{\alpha}(n)$ to finish computing the last instance of
$\tau_{\beta}$ requested before $t_{\alpha}(n)$. More specifically, since the $y(n)$-th instance of $\tau_{\beta}$ is defined to
be the first instance of $\tau_{\beta}$ requested after $t_{\alpha}(n)$, the
$\{y(n)-1\}$-th instance of $\tau_{\beta}$ is the last instance of $\tau_{\beta}$ requested
before $t_{\alpha}(n)$. As a result, $P(n)$ describes the remaining
time needed to finish computing the $\{y(n)-1\}$-th instance of $\tau_{\beta}$
after $t_{\alpha}(n).$  This $P(n)$ will be
another state variable, in addition to $t_{\alpha}(n)$,
$t_{\beta}(n)$, $y(n)$ and $z(n)$. We also introduce an auxiliary
variable $R(n)$ that describes how much time within the interval of
$[t_{\alpha}(n), t_{\alpha}(n+1)]$ can be allocated to compute the
last instance of $\tau_{\beta}$ requested before $t_{\alpha}(n)$, i.e. the $\{y(n)-1\}$-th instance of $\tau_{\beta}$,  .

To determine the value of $R(n)$, we need to
consider two possibilities.
\begin{enumerate}
\item If
\begin{equation}
0<P(n) \leq T_{\alpha}(n)-C_{\alpha}(n),
\end{equation}
which implies that the execution of the $\{y(n)-1\}$-th instance of
$\tau_{\beta}$ is finished within
$[t_{\alpha}(n),t_{\alpha}(n+1)]$, then $R(n)=P(n)$.

\item If
\begin{equation}
P(n) > T_{\alpha}(n)-C_{\alpha}(n),
\end{equation}
which implies that the execution of the $\{y(n)-1\}$-th instance of
$\tau_{\beta}$ cannot be finished before $t_{\alpha}(n+1)$, then the
idle time of processor within the interval of $[t_{\alpha}(n),
t_{\alpha}(n+1)]$ will be allocated to compute $\tau_{\beta}$.
Thus, $R(n)=T_{\alpha}(n)-C_{\alpha}(n)$.
\end{enumerate}

As a summary, we have
\begin{equation}
\label{eq:R} R(n)={\rm min}\{T_{\alpha}(n)-C_{\alpha}(n), P(n)\}.
\end{equation}
We can see that $R(n)$ solely depends on $P(n)$, but $R(n)$ will be
more convenient to use to derive processor utilization waveforms in
Section \ref{CPU}.

Next, we determine the transition rules for $P(n)$. There are two
cases to consider.

{\bf Case 1}: the phase variable $z(n)$ satisfies
\begin{equation}
0\leq z(n)< T_{\alpha}(n)
\end{equation}
which implies that the $y(n)$-th instance of $\tau_{\beta}$ is requested within $[t_{\alpha}(n), t_{\alpha}(n+1)]$, as illustrated by
Fig. \ref{fig:R1}.

\begin{figure}[h]
\centering
\epsfig{file=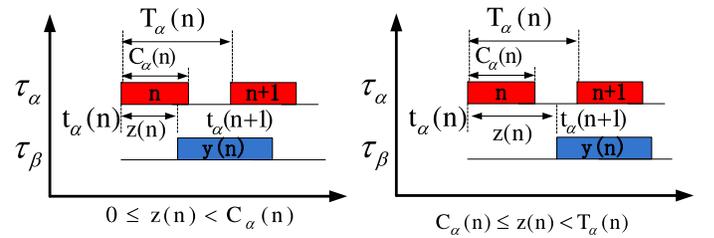, width=0.5\textwidth}
\caption{Case 1. The $y(n)$-th instance of $\tau_{\beta}$ is requested within
$[t_{\alpha}(n),t_{\alpha}(n+1)]$}.  \label{fig:R1}
\end{figure}

In Fig. \ref{fig:R1}, we can see that if $0 \leq z(n)<
C_{\alpha}(n)$, then the execution of the $y(n)$-th instance of $\tau_{\beta}$
will start from $t_{\alpha}(n)+C_{\alpha}(n)$ since the execution of
the $n$-th instance of $\tau_{\alpha}$ will postpone the $y(n)$-th instance of
$\tau_{\beta}$. If $C_{\alpha}(n)\le z(n)<T_{\alpha}(n)$, then the
execution of the $y(n)$-th instance of $\tau_{\beta}$ will start from
$t_{\alpha}(n)+z(n)$. Thus, the execution of the $y(n)$-th instance of $\tau_{\beta}$
will start from
\begin{equation*}
{\rm max} \{t_{\alpha}(n)+C_{\alpha}(n), t_{\alpha}(n)+z(n)\}.
\end{equation*}

The time allocated for the execution of the $y(n)$-th instance of $\tau_{\beta}$
within $[t_{\alpha}(n), t_{\alpha}(n+1)]$ is,
\begin{eqnarray}
&&t_{\alpha}(n+1)-{\rm max} \{t_{\alpha}(n)+C_{\alpha}(n),
t_{\alpha}(n)+z(n)\} \cr &=&T_{\alpha}(n)-{\rm
max}\{C_{\alpha}(n),z(n)\}
\end{eqnarray}
Therefore, the residue for the $y(n)$-th instance of $\tau_{\beta}$ to be
computed after $t_{\alpha}(n+1)$ is
\begin{equation}
C_{\beta}(y(n))-(T_{\alpha}(n)-{\rm max}\{C_{\alpha}(n),z(n)\}).
\end{equation}
 The value of $P(n+1)$ can not be negative. Therefore
\begin{equation}
\label{eq:P1RMS} P(n+1)={\rm
max}\{0,C_{\beta}(y(n))-(T_{\alpha}(n)-{\rm
max}\{C_{\alpha}(n),z(n)\})\}.
\end{equation}

{\bf Case 2}: When the phase variable $z(n)$ satisfies
\begin{equation}
z(n) \ge T_{\alpha}(n)
\end{equation}
which implies that the $y(n)$-th instance of $\tau_{\beta}$ is requested after
$t_{\alpha}(n+1)$, the $\{y(n)-1\}$-th instance of $\tau_{\beta}$ is requested
before $t_{\alpha}(n)$, and no instance of $\tau_{\beta}$ is requested within $[t_{\alpha}(n), t_{\alpha}(n+1)]$, as illustrated
by Fig. \ref{fig:R2} .
\begin{figure}[h]
\centering
\epsfig{file=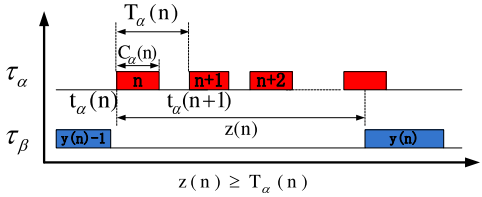, width=0.35\textwidth}
\caption{Case 2. The $y(n)$-th instance of $\tau_{\beta}$ is requested
after $t_{\alpha}(n+1)$} \label{fig:R2}
\end{figure}

In this case, the $\{y(n)-1\}$-th instance of $\tau_{\beta}$ is the last instance of
$\tau_{\beta}$ requested before $t_{\alpha}(n)$. By definition, we
know that $P(n)$ describes the remaining time to compute the
$\{y(n)-1\}$-th instance of $\tau_{\beta}$ after $t_{\alpha}(n)$. In addition,
the $\{y(n)-1\}$-th instance of $\tau_{\beta}$ is also the last instance of $\tau_{\beta}$
requested before $t_{\alpha}(n+1)$ and thus $P(n+1)$ describes the
remaining time to compute the $\{y(n)-1\}$-th instance of $\tau_{\beta}$ after
$t_{\alpha}(n+1)$.

From Fig. {\ref{fig:R2}}, we can see that the time that can be
allocated to compute the $\{y(n)-1\}$-th instance of $\tau_{\beta}$ within
$[t_{\alpha}(n),t_{\alpha}(n+1)]$ is $T_{\alpha}(n)-C_{\alpha}(n)$.
Therefore, the remaining time for the $\{y(n)-1\}$-th instance of $\tau_{\beta}$
to be computed after $t_{\alpha}(n)$ is
\begin{equation*}
P(n)-(T_{\alpha}(n)-C_{\alpha}(n)).
\end{equation*}
Thus the value of $P(n+1)$ is given by the following:
\begin{equation}
\label{eq:P2} P(n+1)={\rm
max}\{0,P(n)-(T_{\alpha}(n)-C_{\alpha}(n))\}.
\end{equation}

According to (\ref{eq:P1RMS}) and (\ref{eq:P2}), we can draw a
conclusion on the transition rules for the state variable $P(n)$ from $n$ to $n+1$.
\begin{equation}
\label{eq:PRMS}
P(n+1)=\left\{
\begin{array}{l}
{\rm max}\{0,C_{\beta}(y(n))-(T_{\alpha}(n)\\
\;\;\;\;\;\;\;\;\;\;\;\;\;\;\;\;\;\;\;\;\;\;\;\;\;\;\;\;\;\;-{\rm max}\{C_{\alpha}(n),z(n)\})\} \\
\;\;\;\;\;\;\;\;\; \mbox{ if } \;\;  0\leq z(n)<T_{\alpha}(n) \;\; \\
{\rm max}\{0,P(n)-(T_{\alpha}(n)-C_{\alpha}(n))\} \\
\;\;\;\;\;\;\;\;\; \mbox{ if } \;\;  z(n)\ge T_{\alpha}(n)\;\;
\end{array}
\right.
\end{equation}


%

\subsection{Initialization of State Variables}

In this subsection, we will discuss the initialization of the state variables within $[t_0,t_0+L]$.

Suppose $t_{\alpha}$ and $t_{\beta}$ is the request time of the first instance of $\tau_{\alpha}$ and $\tau_{\beta}$ requested within $[t_0,t_0+L]$ , $P_{\alpha}$ and $P_{\beta}$ is the residue for the last instance of $\tau_\alpha$ and $\tau_\beta$ requested before $t_0$ to be computed after $t_0$. $P_{\alpha}$ and $P_{\beta}$ need to be retrieved from the memory of the system.

We always assume that $t_{\alpha}(1)=t_0$. If $t_0\neq t_{\alpha}$, we will model a new instance of $\tau_{\alpha}$ as the first instance of $\tau_{\alpha}$ requested within $[t_0,t_0+L]$. The new instance has the following task characteristics:
\begin{equation}
t_{\alpha}(1)=t_0 \;\;\; C_{\alpha}(1)=P_{\alpha} \;\;\;
T_{\alpha}(1)=t_{\alpha}(1)-t_0.
\end{equation}

According to the definition of $t_{\beta}$, $P$, $y$ and $z$, we have
\begin{eqnarray*}
t_{\beta}(1)=t_{\beta}\;\; y(1)=1 \;\; z(1)=t_{\beta}(1)-t_{\alpha}(1) \;\; P(1)=P_{\beta}
\end{eqnarray*}

As we can see, the residue $P$ serves to connect the state of processor before $t_0$ with the state of processor within $[t_0, t_0+L]$. Our model has the advantage over the classical scheduling analysis because instead of storing for all task behaviors made before $t_0$, we only need to record the residue of each task and to start the model. Therefore, we can analysis the scheduled behaviors of tasks starting from any time and within any finite time window as long as Assumptions \ref{assumption:1}  and \ref{assumption:2} are satisfied.

\section{\label{CPU}Processor Utilization Waveform of Two Tasks}
Suppose $\tau_{\alpha}$ and $\tau_{\beta}$ are schedulable. With the
states and transition rules for the scheduling algorithm
determined, we compute a function $\Phi(t)$ that describes the states of the
processor as either free ($\Phi(t)=0$) or busy ($\Phi(t)=1$). To achieve this goal, we study the states of processor within each request interval of $\tau_{\alpha}$, i.e. $[t_{\alpha}(n), t_{\alpha}(n+1)]$. The result is a function of time $\Phi_n$: $t \to \{0,1\}$ where for $t\in [t_{\alpha}(n), t_{\alpha}(n+1)]$,
$\Phi_n(t)=0$ if the processor is free, and $\Phi_n(t)=1$ if the
processor is busy. If $\Phi_n(t)$ is known for all $n$, we are able to describe $\Phi(t)$ within $[t_0, t_0+L]$. Note that we ignore the scheduling transients that may exist for the
processor, then the graph of such a function resembles a square
wave. Although the space of possible square wave functions on a
compact interval is infinite dimensional, we find that for the scheduling policy introduced in Section \ref{formulation}, the waveforms only have two modes.

{\bf Case 1} This describes the mode when the phase variable $z(n)$
satisfies
\begin{equation}
\label{eq:case1RMS} 0 \le z(n)< T_{\alpha}(n),
\end{equation}
which implies that the $y(n)$-th instance of $\tau_{\beta}$ is requested within
$[t_{\alpha}(n),t_{\alpha}(n)+T_{\alpha}(n)]$.

There are two subcases.

{\bf Case 1.1:} If the phase variable $z(n)$ satisfies
\begin{equation}
\label{eq:case11RMS} 0 \leq z(n)<C_{\alpha}(n),
\end{equation}
which implies that the $y(n)$-th instance of $\tau_{\beta}$ is requested while the
processor is computing $\tau_{\alpha}$, i.e. $t_{\beta}(y(n)) \in
[t_{\alpha}(n),t_{\alpha}(n)+C_{\alpha}(n)]$, as illustrated by the
left picture in Fig. \ref{fig:R1}.

It can be observed in Fig. \ref{fig:R1} that
\begin{equation}
\label{eq:phi11RMS} \Phi_n(t)=1 \;\; \mbox{ if } \;\; t \in
\left[t_{\alpha}(n),t_{\alpha}(n)+C_{\alpha}(n)\right].
\end{equation}
The $\{y(n)-1\}$-th instance of $\tau_{\beta}$ must have been finished before
$t_{\alpha}(n)$ as a result of schedulability requirements. Thus,
$R(n)=0$. In order to derive a unified formula for $\Phi_n$ later,
we notice that (\ref{eq:phi11RMS}) can be rewritten as
\begin{equation}
\label{eq:phi12RMS} \Phi_n(t)=1 \;\; \mbox{ if } \;\; t \in
\left[t_{\alpha}(n),t_{\alpha}(n)+C_{\alpha}(n)+R(n)\right].
\end{equation}

The execution of the $y(n)$-th instance of $\tau_{\beta}$ will start
from $t_{\alpha}(n)+C_{\alpha}(n)$. The finishing time of the
$y(n)$-th instance of $\tau_{\beta}$ within $[t_{\alpha}(n), t_{\alpha}(n+1)]$
have two possibilities:
\begin{enumerate}
\item If
\begin{equation}
C_{\beta}(y(n))<T_{\alpha}(n)-C_{\alpha}(n)
\end{equation}
which implies that the execution of $\tau_{\beta}$ will be finished
before $t_{\alpha}(n+1)$, then
\begin{equation}
\begin{array}{c}
\Phi_n(t)=1 \;\; \mbox{ if } \\ t \in
\left[t_{\alpha}(n)+C_{\alpha}(n),t_{\alpha}(n)+C_{\alpha}(n)+C_{\beta}(y(n))\right].
\end{array}
\end{equation}

\item If
\begin{equation}
C_{\beta}(y(n)) \geq T_{\alpha}(n)-C_{\alpha}(n)
\end{equation}
which implies that the execution of $\tau_{\beta}$ will occupy the
idle time of processor within the interval
$[t_{\alpha}(n),t_{\alpha}(n+1)]$, then
\begin{equation}
\begin{array}{c}
\Phi_n(t)=1 \;\; \mbox{ if } \\ t \in
\left[t_{\alpha}(n)+C_{\alpha}(n),t_{\alpha}(n)+C_{\alpha}(n)+(T_{\alpha}(n)-C_{\alpha}(n))\right].
\end{array}
\end{equation}
\end{enumerate}

As a summary, the execution time of the $y(n)$-th instance of $\tau_{\beta}$
during $[t_{\alpha}(n), t_{\alpha}(n+1)]$ is
\begin{equation}
\label{eq:phi13RMS}
\begin{array}{c}
\Phi_n(t)=1 \;\; \mbox{ if } \\t \in
\left[t_{\alpha}(n)+C_{\alpha}(n),t_{\alpha}(n)+C_{\alpha}(n)+\right.\\
\;\;\;\; \left.{\rm
min}\{C_{\beta}(n),T_{\alpha}(n)-C_{\alpha}(n)\}\right].
\end{array}
\end{equation}
According to (\ref{eq:phi12RMS}) and (\ref{eq:phi13RMS}), the
processor utilization during $[t_{\alpha}(n), t_{\alpha}(n+1)]$ is
\begin{equation}
\label{eq:execution11RMS}
\begin{array}{c}
 \Phi_n(t)=1 \mbox{ if } \\
t \in \left[t_{\alpha}(n),t_{\alpha}(n)+C_{\alpha}(n)+R(n)\right]
\bigcup \left[t_{\alpha}(n)+C_{\alpha}(n),\right.\\
\;\; \left.t_{\alpha}(n)+C_{\alpha}(n)+{\rm
min}\{C_{\beta}(n),T_{\alpha}(n)-C_{\alpha}(n)\}\right].
\end{array}
\end{equation}

{\bf Case 1.2:} If the phase variable $z(n)$ satisfies
\begin{equation}
\label{eq:case12RMS} z(n) \ge T_{\alpha}(n)
\end{equation}
which implies that the $y(n)$-th instance of $\tau_{\beta}$ is requested within $[t_{\alpha}(n)+C_{\alpha}(n),t_{\alpha}(n+1)]$ , as
shown in the right picture of Fig. \ref{fig:R2}.

The arguments for the $\{y(n)-1\}$th instance of $\tau_{\beta}$ still holds.
Hence
\begin{equation}
\label{eq:phi14RMS} \Phi_n(t)=1 \;\; \mbox{ if } \;\; t \in
\left[t_{\alpha}(n),t_{\alpha}(n)+C_{\alpha}(n)+R(n)\right].
\end{equation}

The execution of the $y(n)$-th instance of $\tau_{\beta}$ will start from
$t_{\alpha}(n)+z(n)$. The finishing time of the $y(n)$-th instance of $\tau_{\beta}$ within
$[t_{\alpha}(n), t_{\alpha}(n+1)]$ have two possibilities. Using the
similar method for (\ref{eq:phi13RMS}), we can show that the
execution time of the $y(n)$-th instance of $\tau_{\beta}$ during
$[t_{\alpha}(n), t_{\alpha}(n+1)]$ is
\begin{equation}
\label{eq:phi15RMS}
\begin{array}{c}
\Phi_n(t)=1 \;\;\; \mbox{ if } \\
t \in \left[t_{\alpha}(n)+z(n), t_{\alpha}(n)+z(n)+\right.\\ \left.
\;\;\;\;\; {\rm min}\{C_{\beta}(y(n)),T_{\alpha}(n)-z(n)\}\right].
\end{array}
\end{equation}
According to (\ref{eq:phi14RMS}) and (\ref{eq:phi15RMS}), we know
that the processor utilization during
$[t_{\alpha}(n),t_{\alpha}(n+1)]$ is
\begin{equation}
\label{eq:execution12RMS}
\begin{array}{c}
\Phi_n(t)=1 \;\; \mbox{ if } \\
t \in \left[t_{\alpha}(n), t_{\alpha}(n)+C_{\alpha}(n)+R(n)\right]
\bigcup \left[t_{\alpha}(n)+z(n),\right.\\
\left. t_{\alpha}(n)+z(n)+{\rm
min}\{C_{\beta}(y(n)),T_{\alpha}(n)-z(n)\}\right].
\end{array}
\end{equation}

As a summary, a conclusion can be drawn from
(\ref{eq:execution11RMS}) and (\ref{eq:execution12RMS}) that during
$[t_{\alpha}(n), t_{\alpha}(n+1)]$, the processor state is
\begin{equation}
\label{eq:execution1RMS}
\begin{array}{c}
\Phi_n(t)=1 \;\; \mbox{ if } \\
t \in \left[t_{\alpha}(n), t_{\alpha}(n)+C_{\alpha}(n)+R(n)\right]
\bigcup \left[t_{\alpha}(n)+\right.\\\left.{\rm max}\{C_{\alpha}(n),
z(n)\},t_{\alpha}(n)+{\rm max}\{C_{\alpha}(n), z(n)\}\right.\\
\left. +{\rm min}\{C_{\beta}(y(n)),T_{\alpha}(n)-{\rm
max}\{C_{\alpha}(n), z(n)\}\}\right].
\end{array}
\end{equation}

{\bf Case 2} This describes the mode when the phase variable $z(n)$
satisfies
\begin{equation}
\label{eq:case2RMS} z(n) \ge T_{\alpha}(n).
\end{equation}
which implies that the $y(n)$-th instance of $\tau_{\beta}$ is requested after
$t_{\alpha}(n+1)$, as shown in Fig. \ref{fig:R2}.

In Fig. \ref{fig:R2}, only the $\{y(n)-1\}$-th instance of $\tau_{\beta}$
and the $n$-th instance of $\tau_{\alpha}$ will affect the processor state within $[t_{\alpha}(n), t_{\alpha}(n+1)]$. In this case, the
$\{y(n)-1\}$-th instance of $\tau_{\beta}$ is the last instance of
$\tau_{\beta}$ requested before $t_{\alpha}(n)$. Thus, the execution
of the $\{y(n)-1\}$-th instance of $\tau_{\beta}$ and the $n$-th instance of $\tau_{\alpha}$
will take up the time
\begin{equation*}
R(n)+C_{\alpha}(n)
\end{equation*}
which implies that the processor state during
$[t_{\alpha}(n),t_{\alpha}(n+1)]$ is
\begin{equation}
\label{eq:execution2RMS} \Phi_n(t)=1 \;\; \mbox{ if } \;\; t \in
\left[t_{\alpha}(n),t_{\alpha}(n)+C_{\alpha}(n)+R(n)\right].
\end{equation}

According to (\ref{eq:execution1RMS}) and (\ref{eq:execution2RMS}),
the processor utilization within each request interval of $\tau_{\alpha}$
are depicted in Fig. \ref{fig:CPU}.

\begin{figure}[htbp]
\epsfig{file=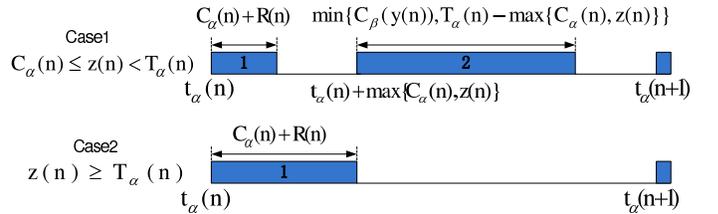,width=0.5\textwidth}
\caption{Processor Utilization Waveform within $[t_{\alpha}(n),
t_{\alpha}(n+1)]$} \label{fig:CPU}
\end{figure}

\section{\label{multiple}Model For Multiple Tasks}
In Section \ref{Model} and \ref{CPU}, we derive the processor
utilization waveform of two tasks from a dynamic scheduling model.
In this section, we develop a recursive algorithm to
extend the modeling effort from two tasks to multiple tasks. Suppose $\tau_1,
..., \tau_N$ are schedulable. The processor utilization waveform of
$\tau_1, ..., \tau_N$ can also be derived recursively in the following
ways

\begin{enumerate}
\item {\bf Initialization}: In the first recursion, we use $\tau_{\alpha}$
to denote $\tau_{1}$ and $\tau_{\beta}$ to denote $\tau_{2}$. Then,
we derive the processor utilization waveform of $\tau_{\alpha}$ and
$\tau_{\beta}$ according to (\ref{eq:execution1RMS}) and
(\ref{eq:execution2RMS}), and model the waveform into a single task
$\tau_{cmb}$.
\item {\bf Iteration}: In the $n$-th iteration ($2\le n \le N-2$), we use $\tau_{\alpha}$
to denote the $\tau_{cmb}$ derived from the $(n-1)$-th iteration and
$\tau_{\beta}$ to denote $\tau_{n+1}$. Then, we model the processor
utilization waveform of $\tau_{\alpha}$ and $\tau_{\beta}$ into a
single task $\tau_{cmb}$.
\item {\bf Result}: In the $\{N-1\}$-th iteration, the processor utilization
waveform of $\tau_{\alpha}$ and $\tau_{\beta}$ is actually the
processor utilization waveform of $\tau_1, ..., \tau_N$
\end{enumerate}

As we can see, the scheduling of multiple tasks are treated
recursively as the scheduling of two tasks. There are two challenges in the recursive algorithm. First, we need to model the
processor utilization waveform as a single task $\tau_{cmb}$; second, we need to
prove that the Assumption \ref{assumption:1} and \ref{assumption:2} are satisfied during each iteration.

First, we show how to characterize $\tau_{cmb}$ as a single schedulable task from the
processor utilization waveform. The processor utilization waveform
within $[t_{\alpha}(n), t_{\alpha}(n+1)]$ have two modes. If $0\le
z(n)<T_{\alpha}(n)$, as illustrated in the upper figure of Fig.
\ref{fig:CPU}, the waveform within $[t_{\alpha}(n),
t_{\alpha}(n+1)]$ is divided into two segments. Each segment can be
viewed as the execution of one instance of $\tau_{cmb}$. The task behaviors of
the first instance is
\begin{equation}
\label{eq:cmb1} \left\{
\begin{array}{l}
{\rm Request \; Interval}: {\rm max}\{C_{\alpha}(n), z(n)\} \\
{\rm Computing \; Time}: C_{\alpha}(n)+R(n)
\end{array}
\right.
\end{equation}
and the task behaviors of the second instance is
\begin{equation}
\label{eq:cmb2} \left\{
\begin{array}{l}
{{\rm Request \; Interval} }: T_{\alpha}(n)-{\rm max}\{C_{\alpha}(n), z(n)\} \\
{\rm Computing \; Time}: \\\;\;\;\;{\rm
min}\{C_{\beta}(y(n)),T_{\alpha}(n)-{\rm max}\{C_{\alpha}(n),
z(n)\}\}
\end{array}
\right.
\end{equation}

On the other hand, if $z(n)\ge T_{\alpha}(n)$, as illustrated in the
lower figure of Fig. \ref{fig:CPU}, the waveforms within
$[t_{\alpha}(n), t_{\alpha}(n+1)]$ can be viewed as the execution of one instance of
$\tau_{cmb}$, with the task behavior being
\begin{equation}
\label{eq:cmb3} \left\{
\begin{array}{l}
{\rm Request \; Interval}: T_{\alpha}(n) \\
{\rm Computing \; Time}: C_{\alpha}(n)+R(n)
\end{array}
\right.
\end{equation}

Next, we use mathematical induction method to prove the
following Proposition.
\begin{proposition}
\label{proposition:1}
In each iteration, for any time instant $t$, there exists $t_0$ and $L$ such that the maximum request interval
of $\tau_{\alpha}$ is smaller than the minimal request interval of $\tau_{\beta}$
within $[t_0, t_0+L]$.
\end{proposition}
\begin{proof}
In the first iteration, $\tau_{\alpha}$ denotes $\tau_{1}$ and
$\tau_{\beta}$ denotes $\tau_{2}$. According to Assumption
\ref{assumption:1} and Assumption \ref{assumption:2}, the Proposition \ref{proposition:1} holds.

Suppose in the $n$-th iteration, for any time instant $t$, there exists $t_0$ and $L$ such that the maximum request interval of
$\tau_{\alpha}$ is smaller than the minimum request interval of
$\tau_{\beta}$ within $[t_0, t_0+L]$. $\tau_{\beta}$ denotes $\tau_{n+1}$. Since $\tau_{cmb}$ is
derived from each request interval of $\tau_{\alpha}$, the maximum
request interval of $\tau_{cmb}$ is no larger than the maximum
request interval of $\tau_{\alpha}$. According to
Assumption \ref{assumption:1}, we know that the maximum
request interval of $\tau_{n+1}$ is smaller than the minimum
request interval of $\tau_{n+2}$ within $[t_0, t_0+L]$. Therefore, the maximum
request interval of $\tau_{cmb}$ is smaller than the minimum
request interval of $\tau_{n+2}$ within $[t_0, t_0+L]$ and Proposition \ref{proposition:1} in the
$\{n+1\}$-th iteration holds.
\end{proof}

\section{\label{simulation}Simulations}
In this section, we verify the dynamic scheduling model by comparing
task scheduling waveforms with those generated by TrueTime. Then, we
simulate the scheduling and the battery models together to predict remaining
battery capacity.

\subsection{Scheduling model verification}
Consider the following scenario:
\begin{enumerate}
\item Initially, two periodic tasks $\tau_1$ and $\tau_2$ are running on a single processor,
with the following parameters: $C_1=0.2 {\rm min}$, $T_1=1 {\rm min}$
$C_2=0.3 {\rm min}$, $T_1=1.5 {\rm min}$. $\tau_1$ starts at the time $0$ while $\tau_2$ starts at the time $0.3{\rm min}$; \item An aperiodic task $\tau_3$ arrives at $50{\rm min}$ and another aperiodic task $\tau_4$ arrives at $50.6{\rm min}$. Both $\tau_3$ and $\tau_4$ will stop after $57.1{\rm min}$. Within $[50.6, 57.1]$, $\tau_3$ and $\tau_4$ have the following characteristics
$T_{3}(1)=1.6{\rm min}$, $C_{3}(1)=0.5{\rm min}$, $T_{3}(2)=2{\rm min}$, $C_{3}(2)=0.6{\rm min}$, $T_{3}(3)=1.7{\rm min}$, $C_{3}(3)=0.2{\rm min}$, $T_{3}(4)=1.8{\rm min}$, $C_{3}(4)=0.4{\rm min}$,
$T_{4}(1)=2.5{\rm min}$, $C_{4}(1)=0.1{\rm min}$, $T_{4}(2)=3{\rm min}$, $C_{4}(2)=0.4{\rm min}$, $T_{4}(3)=1{\rm min}$, $C_{4}(3)=0.3{\rm min}$.
\item At time $110{\rm min}$, $\tau_{2}$ stops. Another periodic task $\tau_5$ arrives at $111.3{\rm min}$, with $T_5=1.2{\rm min}$ and $C_5=0.4{\rm min}$.
\item A aperiodic task $\tau_6$ arrives at time $113{\rm min}$ and disappears after $120{\rm min}$. Within $[113, 120]$, $\tau_6$ has the following characteristics $T_{6}(1)=4{\rm min}$, $T_{6}(2)=3{\rm min}$, $C_{6}(1)=0.6{\rm min}$, $C_{6}(2)=0.3{\rm min}$.
\end{enumerate}

We are interested in the state of processor within $[50, 57.1]$ and $[110,120]$. The waveforms shown in Fig.\ref{fig:scheduling1} and Fig.\ref{fig:scheduling2} are consistent with those generated by Truetime. However, by using truetime, we have to initialize system whenever new tasks arrive and to simulate from the beginning.
\begin{figure}[htbp]
\centering
\epsfig{file=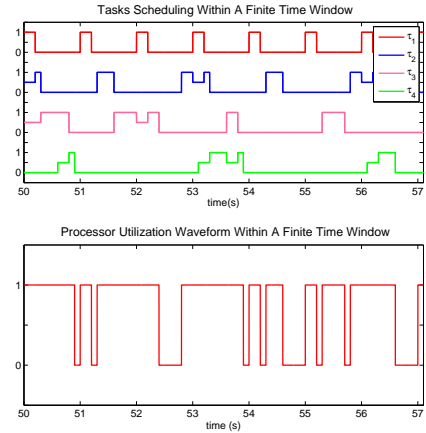, width=0.30\textwidth}
\caption{The State of Processor within $[50, 57.1]$} \label{fig:scheduling1}
\end{figure}

\begin{figure}[htbp]
\centering
\epsfig{file=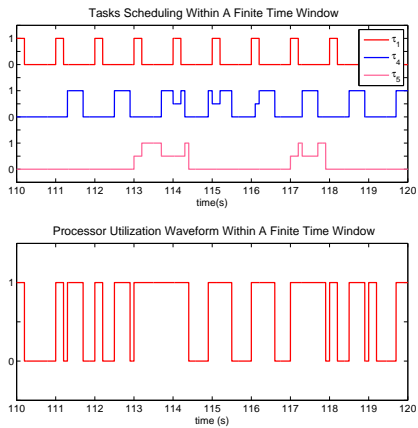, width=0.30\textwidth}
\caption{The State of Processor within $[110, 120]$} \label{fig:scheduling2}
\end{figure}

As we can see in Fig.\ref{fig:scheduling1} and
\ref{fig:scheduling2}, the value of a task will jump from $0$ to
$1$ when it begins to execute. Then, the value of this task might go
through several transitions between $1$ and $0.5$ during the
execution period. Finally, the value of this task will fall back to
$0$ once the execution is finished.

\subsection{SPSB simulation}
Fig. \ref{fig:battery1} shows the variation of  battery capacity
depending on processor status. When processor is busy, i.e.
$\Phi(t)=1$, the battery capacity loss keeps increasing. When
processor is free, i.e. $\Phi(t)=0$, the battery capacity loss
begins to decrease due to the recovery effect.

We simulate a battery same as that in \cite{Rakhmatov03}, which has the parameters $\alpha=40375$ and $
\beta=0.273$.  We assume that the current drawn by the processor when it is busy is
$I=200$ mA and the current vanishes when the processor is free.

Fig. \ref{fig:battery1} shows the state of battery within two time windows $[50, 57.1]$ and $[110, 120]$. According to Assumption \ref{assumption:1} and \ref{assumption:2}, we know that the state of the battery within its whole life period can be monitored by properly shifting the finite time window.

\begin{figure}[htbp]
\centering
\epsfig{file=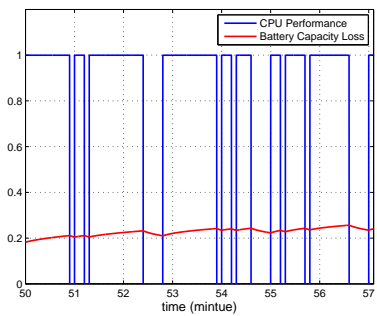, width=0.27\textwidth}
\caption{Simulating the hybrid model within [50 57.1]} \label{fig:battery1}
\end{figure}

\begin{figure}[htbp]
\centering
\epsfig{file=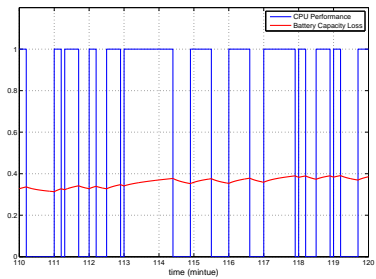, width=0.27\textwidth}
\caption{Simulating the hybrid model within [110 120]} \label{fig:battery2}
\end{figure}

\section{\label{con}Conclusions}
In this paper, we
have established analytical models for the behaviors of multiple aperiodic tasks scheduled on a single
processor supported by a single battery. We assume that the tasks are scheduled under a RMS like algorithm that
assigns priority of tasks based on information within a finite time window.   Our model is presented using a set of nonlinear
difference equations that describe the task behaviors together with continuous differential equations that describe battery discharge behavior.
One may find that these equations can be studied as a hybrid system. It can be perceived that through our modeling effort, battery behavior and
scheduled task behaviors can now be studied jointly within the same theoretical framework. This fact is well aligned with
the goals of cyber physical systems theory.


\end{document}